\documentclass{article}
\usepackage{amsmath,amssymb,amsthm}
\usepackage{latexsym}
\title{On a Hamiltonian form of an elliptic spin Ruijsenaars-Schneider system}
\author{F.Soloviev}

\newtheorem{theorem}{Theorem}

\begin{document}
\maketitle

\section{Introduction}
An elliptic Ruijenaars-Schneider (RS) model~\cite{RS86} is a Hamiltonian system of $N$ interacting particles with a Hamiltonian
\begin{equation}\label{HRS}
H=\sum_{j=1}^N e^{p_j} \prod_{s \ne j}^N \left( \dfrac{\sigma(x_j-x_s+\eta) \sigma(x_j-x_s-\eta)}{\sigma^2(x_j-x_s)} \right)^{1/2}
\end{equation}
and the canonical symplectic form $\omega=\sum \delta p_i \wedge \delta x_i$, where $p_i=\dot{x}_i$.

The equations of motions are
\begin{equation}\label{RS1}
\ddot{x}_i=\sum_{s \ne i} \dot{x}_i \dot{x}_s (V(x_s-x_i)-V(x_i-x_s)),
\end{equation}
where $V(x) = \zeta(x+\eta)-\zeta(x),$ and $\zeta(x)$ is a Weierstrass zeta function.

The limit when one or two periods of the elliptic curve go to infinity yields a trigonometric or rational
system. A RS system is a relativistic generalization of the Calogero-Moser model.

A spin generalization of RS system was suggested in~\cite{KZ95}. Each particle additionally carries
two $l$-dimensional vectors $a_i$ and $b_i$ that describe the internal degrees of freedom and affect the interaction.
Remarkably, the equations of motion remain integrable and are given by the formulas
\begin{equation}\label{RS2}
\begin{cases}
\dot{f}_{ij}=\sum_{k \ne j} f_{ik} f_{kj} V(x_j-x_k) - \sum_{k \ne i} f_{ik} f_{kj} V(x_k-x_i)\\
\dot{x}_i=f_{ii},
\end{cases}
\end{equation}
where $f_{ij} = b_i^T a_j$.

It was shown in~\cite{IK98} using the universal symplectic form (proposed in~\cite{KP97}) that
a spin elliptic RS system is Hamiltonian. An expression of a symplectic form (or Poisson structure)
in explicit coordinates is known only in the rational and trigonometric limits (see~\cite{AF98}).

The aim of this paper is to compute $\omega$ in the original coordinates $x_i$ and $f_{ij}$ in the simplest
elliptic case of 2 particles, $N=2$. We compare the obtained 2-form with a symplectic form for a system without spin
and with a Poisson structure found in~\cite{AF98} in the rational case.

\section{Symplectic form in the case $N=2$}

The general procedure developed by Krichever and Phong in~\cite{KP97} allows to construct action-angle variables for
an elliptic RS system and its spin generalization. It was done in~\cite{IK98}.
The upshot of the procedure is the following.

A Lax representation with a spectral parameter for an elliptic RS system has been found in~\cite{KZ95}.
A Lax matrix is
\begin{equation}\label{laxrs}
L_{ij}=f_i \Phi(x_i-x_j-\eta) \text{, where } \Phi(x,z)=\dfrac{\sigma(z+x+\eta)}{\sigma(z+\eta) \sigma(x)} \left[ \dfrac{\sigma(z-\eta)}{\sigma(z+\eta)} \right]^{x/2\eta}.
\end{equation}
The spectral parameter $z$ is defined on an elliptic curve $\Gamma_0$ with a cut between points $z=\eta$ and $z=-\eta$.

The universal symplectic form is given by the formula
\begin{equation}\label{form}
\omega = -\dfrac{1}{2} \sum_{q \in I} \text{res}_q \thinspace \text{Tr}\left( \Psi^{-1} L^{-1} \delta L \wedge \delta \Psi -
\Psi^{-1} \delta \Psi \wedge K^{-1} \delta K \right) dz,
\end{equation}
where the sum is taken over the poles of $L$ and zeroes of $det \thinspace L$.
$\Psi$ is a matrix composed of eigenvectors of $L$, which has poles $\hat{\gamma}_s$ on the spectral curve $\hat{\Gamma}: det \thinspace (L-k I)=0$
due to normalization of eigenvectors. $k$ is a meromorphic function on $\hat{\Gamma}$ and the
matrix $K=diag(k_1,...,k_N)$ is composed of values of $k$ on different sheets of $\hat{\Gamma}$.

$\omega$ doesn't depend on the gauge transformations $L \to gLg^{-1}$ and the normalization of eigenvectors on the leaves
where the form $\delta \ln{k} dz$ is holomorphic.
$\text{Tr}\left( ... \right) dz$ is a meromorphic differential, and the sum of all its residues is zero.
Using these facts, one can show that on the leaves
\begin{equation}\label{uform}
\omega=\sum_s \delta \ln{k(\hat{\gamma}_s)} \wedge \delta z(\hat{\gamma}_s).
\end{equation}
Computations performed in~\cite{IK98} for Lax matrix~(\ref{laxrs}) show that
\begin{equation}\label{eform}
\omega=\sum_i \delta \ln{f_i} \wedge \delta x_i + \sum_{i \ne j} V(x_i-x_j) \delta x_i \wedge \delta x_j,
\end{equation}
where
$$
f_i=e^{p_i} \prod_{s \ne i}^N \left( \dfrac{\sigma(x_i-x_s+\eta) \sigma(x_i-x_s-\eta)}{\sigma^2(x_i-x_s)} \right)^{1/2}
$$
and the Hamiltonian for system~(\ref{RS1}) is $H=\sum_{i=1}^N f_i$.

A Lax representation with a spectral parameter for an elliptic spin RS system~(\ref{RS2}) has
been found in~\cite{KZ95}.
The Lax matrix is $L_{ij}=f_{ij} \Phi(x_i-x_j-\eta)$.
Formally, equations~(\ref{RS2}) are Hamiltonian with $H=\sum_{i=1}^N f_{ii}$ and symplectic form~(\ref{form})
(see~\cite{IK98} for details).
The goal of this paper is to compute form~(\ref{form}) in the original coordinates $x_i$ and $f_{ij}$.

After the gauge transformation by a diagonal matrix $g=diag(\Phi(x_1,z), \lambda \Phi(x_2,z))$
with an appropriate choice of $\lambda$, the matrix $L_{ij}$ becomes
\[
L=
\begin{pmatrix}
 -f_1 \dfrac{\sigma(z)}{\sigma(\eta)} \qquad f_3 \dfrac{\sigma(z-x_1+x_2) \sigma(z+x_1+\eta) \sigma(x_2)}{\sigma(x_2-x_1-\eta) \sigma(z+x_2+\eta) \sigma(x_1)}\\
f_3 \dfrac{\sigma(z+x_1-x_2) \sigma(z+x_2+\eta) \sigma(x_1)}{\sigma(x_1-x_2-\eta) \sigma(z+x_1+\eta) \sigma(x_2)} \qquad -f_2 \dfrac{\sigma(z)}{\sigma(\eta)}
\end{pmatrix}
\dfrac{1}{\sqrt{\sigma(z+\eta) \sigma(z-\eta)}},
\]
where $f_1 \equiv f_{11}, f_2 \equiv f_{22}$ and $f_3 \equiv \sqrt{f_{12} f_{21}}$.

The matrix $L$ is defined on a curve $\Gamma$ of genus $g=2$, which is a 2-sheeted cover
of the elliptic curve $\Gamma_0$ with 2 branch points $z=\eta$ and $z=-\eta$.

The spectral curve $\hat{\Gamma}$ of $L$ is defined by the equation $R=det \thinspace (L_{ij}-k)=0$.
It is a 2-sheeted cover of $\Gamma$, and the function $\partial_k R$ has 4 simple poles on $\hat{\Gamma}$
above points $z=\pm\eta$. $\partial_k R$ is a meromorphic function on $\hat{\Gamma}$, hence
it also has 4 zeroes. Its zeroes are precisely the branch points of $\hat{\Gamma}$ over $\Gamma$,
and the Riemann-Hurwitz formula implies that the genus of $\hat{\Gamma}$ is $\hat{g}=5$.

The matrix valued differential $L dz$ can be seen as a global
section of the bundle $End(V_{\gamma,\alpha}) \otimes \Omega^{1,0}(\Gamma)$.
$V_{\gamma,\alpha}$ is a vector bundle determined by Tyurin parameters $z(\gamma_i)=-x_1-\eta$,
$z(\gamma_j)=-x_2-\eta$, and $\alpha_i=(0,1)^T$, $\alpha_j=(1,0)^T$, where $i=1,2$ and $j=3,4$.

The set $I$ in~(\ref{form}) is $I=\{\gamma_s,0,\pm z_0\}$, where $z_0$ is defined by the equation $det \thinspace L(z_0)=0$, or
$$
f_1 f_2 \dfrac{\sigma^2(z_0)}{\sigma^2(\eta)} - f_3^2 \dfrac{\sigma(z_0+x_1-x_2) \sigma(z_0-x_1+x_2)}{\sigma(x_1-x_2-\eta) \sigma(x_2-x_1-\eta)}=0.
$$
Notice, that we can use variables $(x_1,x_2,f_1,f_2,z_0)$ instead of $(x_1,x_2,f_1,f_2,f_3)$.

\begin{theorem}\label{th1}
In the case $N=2$ the elliptic spin RS system is Hamiltonian with a symplectic form
\begin{equation}\label{newform}
\omega=-\delta \ln{f_1} \wedge \delta x_1 - \delta \ln{f_2} \wedge \delta x_2+
2 \tilde{V}(x_1-x_2) \delta x_1 \wedge \delta x_2
\end{equation}
and Hamiltonian $H=f_1+f_2$, where $\tilde{V}(x)=\zeta(x+z_0)-\zeta(x)$.
The spinless case corresponds to $z_0=\eta$.
\end{theorem}

\begin{proof}
The eigenvector $\psi$ of $L$ in any normalization is a meromorphic function on $\hat{\Gamma}$
and it has $\hat{g}+1=6$ poles $\hat{\gamma}_s$.
The proof of formula~(\ref{uform}) in~\cite{IK04} assumes that the situation is in general
position, i.e. projections of points $\hat{\gamma}_i$ don't coincide with $\gamma_s$.

Most appropriate normalization here is $\psi_1 \equiv 1$, because it easily allows us to find poles
$\hat{\gamma}_s$ of $\psi$. Two of them ($s=1,2$) lie above the point $z=x_1-x_2$, and the other
are above $z=-x_1-\eta$ ($s=3,4,5,6$). This is not the case of general position, but it
turns out that the same formula~(\ref{uform}) still holds.

The proof in~\cite{IK02} and~\cite{IK04} implies that 2-form~(\ref{form}) in the
normalization $\psi_1 \equiv 1$ equals to
$\omega_0=\sum_{s=1}^2 \delta \ln{k(\hat{\gamma}_s)} \wedge \delta z(\hat{\gamma}_s).$

A change of normalization of $\Psi$ from $\psi_1 \equiv 1$ to $\sum \psi_i \equiv 1$ (the last one
is in "general position")
corresponds to the transformation $\tilde{\Psi}=\Psi V$, where
\[
V=
\begin{pmatrix}
\dfrac{L_{12}}{k_1-L_{11}+L_{12}} && 0\\
0 && \dfrac{L_{12}}{k_2-L_{11}+L_{12}}
\end{pmatrix}.
\]

According to the computations in~\cite{IK04},
$$
\omega=\omega_0+\sum_{q \in I} \text{res}_q \thinspace \text{Tr}\left(K^{-1} \delta K \wedge \delta V V^{-1} \right) dz.
$$
Since $\omega$ has to be restricted to the leaves where $\delta \ln{k} dz$ is holomorphic (which
is equivalent to 2 conditions: $\delta \eta=0$ and $\delta z_0=0$),
the only non-zero residue in the second term is at the point $z(\gamma_i)=-x_1-\eta$.
After computing the residue, we get that
$\omega=\omega_0+\sum_{s=3}^6 \delta \ln{k(\hat{\gamma}_s)} \wedge \delta z(\hat{\gamma}_s)$,
i.e. effectively formula~(\ref{uform}) holds in both normalizations.

Substituting $\hat{\gamma}_s$ in~(\ref{uform}), we find that
$$
\omega=-\delta \ln{f_1} \wedge \delta x_1-\delta \ln{f_2} \wedge \delta x_2+
2 \tilde{V}(x_1-x_2) \delta x_1 \wedge \delta x_2,$$
where $\tilde{V}(x)=\zeta(x+z_0)-\zeta(x)$.

The Hamiltonian $H=f_1+f_2$ defines the flow
\[
\begin{cases}
\dot{f}_1=-f_1 f_2 (\zeta(z_0+x_1-x_2)-\zeta(z_0-x_1+x_2)-2\zeta(x_1-x_2))\\
\dot{f}_2=f_1 f_2 (\zeta(z_0+x_1-x_2)-\zeta(z_0-x_1+x_2)-2\zeta(x_1-x_2))\\
\dot{x}_1=f_1\\
\dot{x}_2=f_2.
\end{cases}
\]

Using identities for Weierstrass $\sigma$-functions, namely,
$$
\sigma(a+c)\sigma(a-c)\sigma(b+d)\sigma(b-d)-\sigma(a+d)\sigma(a-d)\sigma(b+c)\sigma(b-c)=
$$
$$
=\sigma(a+b)\sigma(a-b)\sigma(c+d)\sigma(c-d) \text{, and}$$
$$
\zeta(a)+\zeta(b)+\zeta(c)-\zeta(a+b+c)=\dfrac{\sigma(a+b)\sigma(b+c)\sigma(a+c)}{\sigma(a)\sigma(b)\sigma(c)\sigma(a+b+c)},
$$
it follows from the definition of $z_0$ that
\begin{equation}\label{imprel}
f_1 f_2 (2\zeta(x_1-x_2)+\zeta(z_0-x_1+x_2)-\zeta(z_0+x_1-x_2))=
\end{equation}
$$
=f_3^2 (2\zeta(x_1-x_2)+\zeta(\eta-x_1+x_2)-\zeta(\eta+x_1-x_2)).
$$
With the help of this identity, we can show that the above equations are equivalent to
\[
\begin{cases}
\ddot{x}_1 = f_3^2 (2\zeta(x_1-x_2)+\zeta(\eta-x_1+x_2)-\zeta(\eta+x_1-x_2))\\
\ddot{x}_2 = -f_3^2 (2\zeta(x_1-x_2)+\zeta(\eta-x_1+x_2)-\zeta(\eta+x_1-x_2)),
\end{cases}
\]
which is an RS system.

The spinless case occurs when $f_3^2 = f_1 f_2$ and $z_0 = \eta$ as one can observe from~(\ref{imprel}).
\end{proof}

\emph{Remark.} A Poisson structure was found in~\cite{AF98} in the rational limit for arbitrary $N$ (see formula (3.31) in~\cite{AF98}).
In the case of 2 particles it is non-degenerate and defined on a 6-dimensional space
$(f_{11},f_{12},f_{21},f_{22},x_1,x_2)$. The corresponding 2-form is defined on the same space and
coincides with~(\ref{newform}) on the leaves $\delta z_0=0$ and after reduction with respect to the
action $f_{12} \to f_{12}/\lambda \text{, } f_{21} \to f_{21} \lambda$.

\section{Acknowledgments}

I am very grateful to I.Krichever for many helpful and interesting discussions.

\end{document}